\documentclass[10pt, final, journal, letterpaper, twocolumn]{IEEEtran}
\usepackage{amssymb}
\usepackage{amsmath}
\usepackage{amsfonts}
\usepackage{graphicx}
\usepackage{algorithm}
\usepackage{algorithmic}
\usepackage{epstopdf}
\usepackage{cite}
\usepackage{exscale}
\usepackage{relsize}
\usepackage{stfloats}

\usepackage{color}
\usepackage{multirow}
\usepackage{subfigure}
\usepackage{bm}

\newtheorem{theorem}{\textbf{Theorem}}

\newtheorem{proposition}{\textbf{Proposition}}
\newtheorem{remark}{\textbf{Remark}}

\begin{document}
\title{Interference Cancellation at Receivers in Cache-Enabled Wireless Networks}
\author{Chenchen Yang$^{*}$, Yao Yao$^{\dag}$, Bin Xia$^{*}$, Kaibin Huang$^{\divideontimes}$, Weiliang Xie$^{\ddag}$, Yong Zhao$^{\ddag}$\\
$^{*}$ Department of Electronic Engineering, Shanghai Jiao Tong University, Shanghai, P. R. China\\
$^{\divideontimes}$ Department of Electrical and Electronic Engineering, the University of Hong Kong\\
$^{\dag}$ Huawei Technologies Co., Ltd;~
$^{\ddag}$ Technology Innovation Center, China Telecom\\
Email: {\{zhanchifeixiang, bxia\}@sjtu.edu.cn, \{yyao, huangkb\}@eee.hku.hk, \{xiewl, zhaoyong\}@ctbri.com.cn}
}
\maketitle
\begin{abstract}
In this paper, we propose to exploit the limited cache packets as side information to cancel incoming interference at the receiver side. We consider a stochastic network where the random locations of base stations and users are modeled using Poisson point processes.  Caching schemes to reap both the local caching gain and the interference cancellation gain for the users are developed based on two factors: the density of different user subsets and the packets cached in the corresponding  subsets. The packet loss rate (PLR) is analyzed, which depends on both the cached packets and the channel state information (CSI) available at the receiver. Theoretical results reveal the tradeoff between caching resource and wireless resource. The performance for different caching schemes are analyzed and the minimum achievable PLR for the distributed caching is derived.
\end{abstract}
\section{INTRODUCTION}
Both the wireless network topology and the information transmission mode are changing with the advancement of information and communication technologies. In regard to the wireless network topology, nodes are densely and randomly located, yielding serious interference \cite{SINR}. Interference is a key limitation on future networks. In regard to information transmission, content-centric services (e.g., multimedia transmission) instead of connection-centric services (e.g., voice communication) account for most of mobile traffic \cite{5G2},\cite{lau5G}.
Caching exploiting content-centric traffic has been proposed to unleash the ultimate potential of the network \cite{femto2, caching1}.

There are some works on interference management in cache-enabled networks, which are in terms of the degrees of freedom (DoF) from the information-theoretic point of view. 
In \cite{aided}, an interference channel with three transmitters and three receivers is considered. It is shown that caching split files at transmitters can increase the throughput via interference management. In \cite{IA}, the DoF gain is proved to be achievable via caching parts of files at the transmitters for a three-user interference channel. Caching at base stations (BSs) for opportunistic multiple-input multiple-output cooperation  is proposed in \cite{Anliu} to achieve the DoF gain without requiring high-speed fronthaul links. In \cite{lau5G}, physical layer (PHY)-caching at BSs is proposed to mitigate interference and improve the number of DoF in wireless networks. The storage-latency tradeoff  is analyzed in \cite{sengupta} for the network with cache-enabled transmitters, and the transmission rate is specified by the DoF. The standard DoF is adopted in \cite{tranrecei} as the performance metric for the network with interference channels. Transmitters and receivers are with caching strategy of equal file splitting. \cite{hachem} studies the benefit of caching for the system with  two cache-enabled transmitters and two cache-enabled receivers in the interference channel. The layered architecture is proposed and the DoF for the optimal strategy is computed.  The concept of fractional delivery time (FDT) is proposed in \cite{TMM} to reflect  the DoF enhancement due to transmitter caching and the load reduction due to receiver caching.  A complete constant-factor approximation of the DoF is proposed in \cite{hachem2} for the network with caching at both transmitters and receivers.

However,  previous works with  information-theoretic framework assume that the global channel state information (CSI) is available. The performance for the network with only partial CSI available should be investigated.
The randomness and complexity of node locations due to the stochastic topology of the network need to be addressed.
And interference management is performed based on the caching at the transmitter side in previous works, where extra payload of fronthaul/backhaul is needed for the cooperation among transmitters. Moreover, caching schemes for cache-enabled networks to exploit both  the local caching gain and the interference cancellation gain at the receiver side need to be elaborated further. Different from the information-theoretic framework, in this paper, we focus on the stochastic network where random numbers of BSs and users are spatially located in the two-dimensional plane. The effects of CSI are addressed and the caching schemes are analyzed. Our main contributions are summarized as follows,
\begin{itemize}
\item We propose to cancel the incoming interference with  partial CSI and cached packets at users. Random numbers of BSs and users are considered in the stochastic network.
\item The effects of the CSI are analyzed on the network performance in terms of packet loss rate (PLR), specifically, when partial, global and none of CSI are available.
\item The effects of the caching scheme on the PLR are further elaborated. And the optimal caching scheme for the users with distributed caching are provided to reap both the local caching gain and the interference cancellation gain.
\end{itemize}

\section{System model and Protocol description}\label{sec:system}
\subsection{Network Architecture}
Consider the wireless network where BSs and users are independently located according to Poisson Point Processes (PPPs) $\Phi_b$ and $\Phi_u$. The intensities of  $\Phi_b$ and $\Phi_u$ are $\lambda_b$  and $\lambda_u$, respectively. The system is slotted and the duration of each slot is $\tau$ $($seconds$)$. Each user is assumed to randomly request a packet of the fixed length $T$ $($Mb$)$ in a slot from the packet library $\mathcal{L} \triangleq \{l_1, l_2, \cdots ,l_N\}$ \cite{TMM}. The packet $l_i, i\in\{1, 2, \cdots,N\}$ is requested by the user with the probability $f_i\in[0,1]$ in the slot. Define  $\mathcal{F} \triangleq \{f_1, f_2, \cdots ,f_N\}$ as the packet access probability set. Note that $\sum_{i=1}^N f_i=1$ and  without loss of generality (w.l.o.g.), $f_1 \geq f_2,\cdots, \geq f_N$.

Each user has a limited caching storage with size of $M\times T$ $($Mb$)$, and $M$ of the $N$ packets in the packet library $\mathcal{L}$ are pre-cached at the user. Therefore, there are totally $ \binom{N}{M}$ kinds of caching schemes for different users.  Classify all the users  into $\binom{N}{M}$ subsets according to their caching schemes.  Denote the density of  the users in the $i$-th  subset as $p_i\lambda_u$ for $p_i\in[0,1]$, $i=1,2,\cdots,\binom{N}{M}$ and  $\sum_{i=1}^{\binom{N}{M}} p_i=1$. W.l.o.g., assume $p_1\geq p_2,\cdots,\geq p_{\binom{N}{M}}$.  Let $q_{i,j}\in\{0, 1\}$ denote whether users in the $i$-th subset has cached packet $l_j$ ($j=1,2,\cdots,N$), where  $q_{i,j}=1$ indicates that packet $l_j$  has been cached in the users of the $i$-th subset,  and $q_{i,j}=0$ otherwise.  Then the caching scheme for the network depends on matrices $\mathbf{P}\triangleq[p_i]_{N\times 1}$ and $\mathbf{Q}\triangleq[q_{i,j}]_{\binom{N}{M}\times N}$. 
Denote the set of the cached packets in the users of the $i$-th subsect as $\mathcal{L}_i\triangleq\{l_j: q_{i,j}=1, j=1,2,\cdots,N\}$. When the requested packet has been cached at the user, the user reads it out immediately from its local caching; otherwise, the user should obtain the requested packet from its nearest BS. Requests are waiting to be served in the infinite buffer of the BS, and each BS is assumed to transmit a packet in each slot on the FIFO-basis (first-in, first-out). Consider the service discipline that a request is dropped out of the buffer at the end of the slot assigned to it, no matter whether the BS has transmitted the requested packet successfully or not. The packet loss rate will be elaborated in Section \ref{sec:lossrate}.
\subsection{Interference Cancellation with Cached Packets}
\begin{figure}[t]
\centering
\includegraphics[width=2.8in]{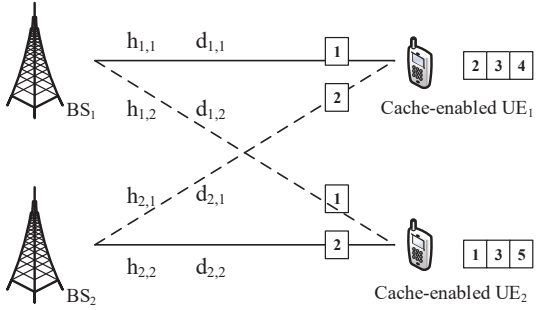}
\caption{Interference cancellation in the cache-enabled wireless network.}
\label{interference}
\end{figure}
All BSs share the same wireless channel with bandwidth of $B$ $($MHz$)$ to transmit packets to users in the downlink. Consider the CSI in the downlink is available at the user, if the distance from the BS to the user is smaller than $R_b$.
As a simplified prototype of the network, Fig.  \ref{interference} illustrates two BSs ($\text{BS}_1$ and $\text{BS}_2$) and two cache-enabled users. The first cache-enabled user has stored packets $l_2,l_3,l_4$ and the second cache-enabled user has stored packets $l_1, l_3, l_5$. The two users are covered by $\text{BS}_1$ and $\text{BS}_2$, respectively. When  $\text{BS}_1$  and $\text{BS}_2$ transmit packet $l_1$ and $l_2$, respectively, to the first and second users at the same slot, the received signal of the $i$-th  user is
 \begin{equation} \label{yi}
    y_i=\sqrt{P}d_{i,i}^{-\frac{\beta}{2}}h_{i,i}x_i+\sqrt{P}d_{j,i}^{-\frac{\beta}{2}}h_{j,i}x_j+n_0,
 \end{equation}
for $i\neq j\in\{1, 2\}$. $P$ is the transmit power,  $d_{i,i}$ and $d_{j,i}$  are respectively the distance from the serving $\text{BS}_i$ and the  interfering $\text{BS}_j$ to the $i$-th user, $h_{i,i}$ and $h_{j,i}$  are the corresponding channel fading, $\beta\geq 2$ denotes the path-loss exponent, $x_i$ and $x_j$ are the transmit signal with unit power, $n_0\sim \mathcal{CN}(0,\sigma^2)$ is the zero-mean additive white Gaussian noise (AWGN) with power $\sigma^2$. When the distance $d_{j,i}$ from $\text{BS}_j$ to the $i$-th user is smaller than $R_b$,  the CSI $\sqrt{P}d_{j,i}^{-\frac{\beta}{2}}h_{j,i}$ is known to the $i$-th user. Consider there is a specific index table of the packets and it is known to
all BSs and users. Before the interfering $\text{BS}_j$ transmits a packet (e.g., packet $l_j$) which has been cached by the users around, $\text{BS}_j$ broadcasts the packet index  via extra interactive signals to let the users around know the incoming interference signal $x_j$ in the slot. Therefore, based on the CSI  knowledge  (i.e., $\sqrt{P}h_{j,i}d_{j,i}^{-\frac{\beta}{2}}$) and the side information (i.e., the cached packet $l_j$: $x_j$), the $i$-th user can cancel the interference (i.e., the term ``$\sqrt{P}d_{j,i}^{-\frac{\beta}{2}}h_{j,i}x_j$" in equation (\ref{yi})). Accordingly,  the signal-to-interference-plus-noise ratio (SINR) of the $i$-th user is $\text{SINR}_i={P|h_{i,i}|^2d_{i,i}^{-\beta}}{\sigma^{-2}}$. 
\section{The packet loss rate}\label{sec:lossrate}
In each slot, $(1-q_{i,j})f_j$ of users in the $i$-th subset sent their requests to BSs to obtain packet $l_j$. Therefore,  in the plane, totally $\sum_{i=1}^{\binom{N}{M}}p_i(1-q_{i,j})f_j$ of users sent their requests to the BSs to obtain packet $l_j$. Then based on the aforementioned transmission scheme, in each slot the fraction of BSs transmitting packet $l_j$ is $\alpha_j$, which can be calculated by
 \begin{align}\label{alphaj}
\alpha_j=\frac{\sum_{i=1}^{ \binom{N}{M}}p_i(1-q_{i,j})f_j}{\sum_{j=1}^N\sum_{i=1}^{ \binom{N}{M}}p_i(1-q_{i,j})f_j}\triangleq\frac{f_j-q_jf_j}{1-\sum_{j=1}^Nq_jf_j},
 \end{align}
where $q_j\triangleq\sum_{i=1}^{ \binom{N}{M}}p_iq_{i,j}$ for $j=1,2,\cdots, N$. It is the average caching probability of packet $l_j$ over the whole network.

Based on  Slivnyak's theorem, we conduct the analysis with considering  that there is a typical user $u_0$ at the origin of the Euclidean area \cite{SINR}. Denote $R$ as the distance between the typical user and its serving BS.  Due to the randomness of the BSs and users, the distance $R$ between the typical user and its serving BS is variable. The probability density function (PDF) of $R$ is $ f_{R}(r)=2\pi \lambda_b r e^{-\pi\lambda_b r^2}$ \cite{SINR}. When the typical user is in the $i$-th subset and it requests packet $l_l\in\mathcal{L}\backslash\mathcal{L}_i$ which has not been cached in the local caching, the received signal of the typical user from its serving BS is given by
\begin{align}
    y_{0,i}&=\sqrt{P}d_{0,0}^{-\frac{\beta}{2}}h_{0,0}x_0+\sum\limits_{j\in\Phi_{b1}\odot \max\{R,R_b\} }\sqrt{P}d_{j,0}^{-\frac{\beta}{2}}h_{j,0}x_j\nonumber\\
    &+\sum\limits_{k\in\Phi_{b2}\odot R }\sqrt{P}d_{k,0}^{-\frac{\beta}{2}}h_{k,0}x_k+n_0,
 \end{align}
where $x_0$ is the desired signal of the  typical user $u_0$ for packet $l_l$. $x_j$ and $x_k$ are interfering signals. $d_{0,0}$ is the distance between the typical user and its serving BS. $d_{j,0}$ and $d_{k,0}$  are the distances from interfering BSs to the typical user. $h_{0,0}$, $h_{j,0}$ and $h_{k,0}$   are the corresponding channel fading. Consider Rayleigh fading channel with average unit power in this paper. $\Phi_{b1}$ and $\Phi_{b2}$ are independently thinning PPPs with parameter
 $\lambda_b\sum_{j=1}^N\alpha_jq_{i,j}$ and $\lambda_b(1-\sum_{j=1}^N\alpha_jq_{i,j})$, respectively. $\Phi_{b1}$  is the distribution of interfering BSs transmitting packets included in the cached packet set $\mathcal{L}_i$, and $\Phi_{b2}$ is that of interfering BSs transmitting packets included in the complementary set $\mathcal{L}\backslash\mathcal{L}_i$. Note that the interference from the interfering BSs (distributed with $\Phi_{b1}$)  inside the circle (centered at the origin with radius $R_b$) can be cancelled by the typical user with the knowledge of  the CSI and the cached packets. Therefore, the residual interference of the typical user comes from \expandafter{\romannumeral1}): the BSs distributed with $\Phi_{b1}$ outside the circle centered at the origin with radius $\max\{R,R_b\} $ (denoted by $\Phi_{b1} \odot \max\{R,R_b\} $), and  \expandafter{\romannumeral2}): the BSs distributed with $\Phi_{b2}$ outside the circle centered at the origin with radius $R$ (denoted by $\Phi_{b2} \odot R$). It can be observed that $\lambda_b\sum_{j=1}^N\alpha_jq_{i,j}=0$, if all the users are with the same caching scheme (i.e., $q_{i,j}\equiv q_{k,j}\in\{0, 1\}, \forall i\neq k\in\{1,2,\cdots,\binom{N}{M}\}$). It implies the network cannot reap any interference cancellation gain when all users are with the same caching scheme. The impacts of caching schemes on the network performance will be investigated in Section \ref{sec:schemes}.

Therefore, the SINR of the typical user $u_0$  is given by
\begin{align}\label{sinr}
    &\frac{P|h_{0,0}|^2d_{0,0}^{-\beta}}{\sum\limits_{j\in\Phi_{b1}\odot \max\{R,R_b\}} {P|h_{j,0}|^2d_{j,0}^{-\beta}}+\sum\limits_{k\in\Phi_{b2}\odot R }{P|h_{k,0}|^2d_{k,0}^{-\beta}}+\sigma^2}\nonumber\\
    &\triangleq\frac{P|h_{0,0}|^2d_{0,0}^{-\beta}}{I_c+I_{u}+\sigma^2}\triangleq\frac{P|h_{0,0}|^2d_{0,0}^{-\beta}}{\bar{I}_c},
\end{align}
where $I_c+I_{u}$ are the residual interference after the interference cancellation.  A packet will be lost if it cannot be transmitted completely over the slot assigned to it. The PLR $\mathcal{P}_l$ can be calculated by
 \begin{align}\label{lossdefine}
   \mathcal{P}_l&\triangleq\mathbb{E}\Big[\mathbb{P}\Big[{\tau}B \text{log}_2(1+\text{SINR})< T\Big]\Big]\nonumber\\
   &=\mathbb{E}\Big[\mathbb{P}\Big(\text{SINR}< 2^{\frac{T}{\tau B}}-1\Big)\Big].
\end{align}
Denote $\bar{T}\triangleq2^{\frac{T}{\tau B}}-1$. The average is taken over both the channel fading distribution and the spatial PPP. We then have,
\begin{theorem}\label{theorem1}
The PLR for users in the $i$-th subset to obtain the un-cached packets with partial CSI  can be calculated with equation (\ref{long}) at the bottom of the next page.
\begin{figure*}[hb]
 \hrulefill
 \begin{align}\label{long}
\mathcal{P}_l&\triangleq1-\mathcal{P}_{s,i}-\mathcal{P}_{b,i}=1-2\pi \lambda_b\!\!\int_0^{R_b}\!\!\! r \text{exp}\Big\{\!-\!r^\beta P^{-1}\bar{T}\sigma^2-
\pi r^2\Big\{\lambda_{b1}\mathcal{Z}_1\Big[\Big(\frac{r}{R_b}\Big)^\beta\bar{T}\Big]+\lambda_{b2}\mathcal{Z}_1(\bar{T})+\lambda_b\Big\}\!\!\Big\}\mathrm{d}{r}\nonumber\\
&-2\pi \lambda_b\!\!\int_{R_b}^\infty\!\!\! r \text{exp}\Big\{\!-\!r^\beta P^{-1}\bar{T}\sigma^2{\setlength\arraycolsep{0.5pt}-}
\pi r^2\lambda_b\left[Z_1(\bar{T})\!+\!1\right]\!\!\Big\}\mathrm{d}{r}.
 \end{align}
\end{figure*}
  \end{theorem}
\begin{proof}
Please refer to the Appendix.
\end{proof}
 \remark
 The PLR  for the un-cached packet deceases with the increase of $R_b$. It can be proved that $\frac{\partial \mathcal{P}}{\partial R_b}<0$ when $R_b<+\infty$. Increasing  $R_b$ (more CSI) is helpful to reap the interference cancellation gain and to reduce the PLR.

Moreover, the PLR of cached packets is zero because cached packets can be read out immediately from local caching.  Accordingly, when the user in the $i$-th subset  requests packet $l_j$ which has or not been cached, the PLR is
\begin{align}
\mathcal{P}_{i,j}=(1-q_{i,j})\mathcal{P}_{l}.
\end{align}
 \remark
For $R_b\rightarrow 0$, the system degenerates to the cache-enabled network without any CSI, making interference cancellation infeasible. The PLR is
 \begin{align}\label{without0}
&\mathcal{P}_{i,j}=
(1-q_{i,j})\Big\{1-\nonumber\\
&2\pi \lambda_b\!\!\int_{0}^\infty\!\!\! r \text{exp}\Big\{\!-\!r^\beta P^{-1}\bar{T}\sigma^2{\setlength\arraycolsep{0.5pt}-}
\pi r^2\lambda_b\left[Z_1(\bar{T})\!+\!1\right]\!\!\Big\}\mathrm{d}{r}\Big\}.
\end{align}

In comparison, the reduction of PLR for  the network with interference cancellation is $\mathcal{P}_{i,j|R_b}-\mathcal{P}_{i,j|R_b\rightarrow0}$.
Furthermore, if the noise is relative small than the interference, i.e., in the interference-limited network ($\sigma^2\!\rightarrow\!0$), the PLR in (\ref{without0}) turns to
\begin{align}\label{without}
&\mathcal{P}_{i,j}=
\frac{1-q_{i,j}}{Z_1^{-1}(\bar{T})+1}.
\end{align}
\remark
For the interference-limited network ($\sigma^2\rightarrow0$) and $R_b\rightarrow\infty$, i.e., when the global interfering CSI  are available at the cache-enabled user (yielding the complete interference cancellation), the PLR is
\begin{align}\label{with}
\!\!\!\mathcal{P}_{i,j}&\!=\!
(1\!-\!q_{i,j})\Big(1\!-\!2\pi \lambda_b\!\!\int_0^{\infty}\!\!\!\! r \text{exp}\Big\{\!-\!\pi r^2\Big[\lambda_{b2}\mathcal{Z}_1(\bar{T})\!+\!\lambda_b\Big]\Big\}\!\mathrm{d}{r}\Big)\nonumber\\
&=\frac{1-q_{i,j}}{1+(1-\sum_{j=1}^N\alpha_jq_{i,j})^{-1}\mathcal{Z}_1^{-1}(\bar{T})}.
\end{align}

 Compared with the PLR for the un-cached packet in (\ref{without}), the PLR for the un-cached packet in  (\ref{with}) decreases owing to the performance gain from the interference cancellation.
\section{The impacts of caching schemes}\label{sec:schemes}
In the previous section, we have derived the PLR $\mathcal{P}_{i,j}$ for the user in the $i$-th subset to obtain packet $l_j$. Then the average PLR in terms of all users over the whole network is
\begin{align}\label{average}
\mathcal{P}=\sum\limits_{i=1}^{ \binom{N}{M}}p_i\sum\limits_{j=1}^N\mathcal{P}_{i,j}f_j.
\end{align}
In this section, we analyze the impact of caching schemes on the average PLR. To gain further insight, we focus on the interference-limited network ($\sigma^2\rightarrow0$) with global CSI for complete interference cancellation ($R_b\rightarrow\infty$). 
Consider users are with distributed caching schemes, i.e., users  are without centralized controller and they have no knowledge of the popularity of the packet in advance. So different packets are cached by users in the network randomly with equal probability (i.e., $q_j=\frac{M}{N}$). We have the following proposition,
\begin{proposition}
For $R_b\rightarrow\infty$, $\sigma^2\rightarrow0$ and $q_j=\frac{M}{N}$, when all the packets have the same access popularity (i.e.,  $f_j=f_k, \forall j\neq k \in \{1,2,...,N\}$), the average PLR is
\begin{align}\label{t1}
\mathcal{P}=\frac{(1-\frac{M}{N})^2}{[1+\mathcal{Z}^{-1}_1(\bar{T})]-\frac{M}{N}}.
\end{align}
\end{proposition}
\begin{proof}
Substituting  $q_j=\frac{M}{N}$ and $f_j=\frac{1}{N}$ to (\ref{alphaj}), we have $\alpha_j=\frac{1}{N}$.  Base on equations (\ref{with}) and (\ref{average}), we then have $\mathcal{P}=\sum_{i=1}^{\binom{N}{M}}p_i\frac{\sum_{j=1}^N(1-q_{i,j})f_j}{1+(1-\sum_{j=1}^N\alpha_jq_{i,j})^{-1}\mathcal{Z}^{-1}_1(\bar{T})}
\stackrel{(a)}{=}\frac{(N-M)^2}{N^2-NM+N^2\mathcal{Z}^{-1}_1(\bar{T})}$.
 Step (a) is obtained by noting that $\sum\nolimits_{i=1}^{ \binom{N}{M}}q_{i,j}p_i=q_{j}=\frac{M}{N}$ and $\sum_{j=1}^Nq_{i,j}=M$. Then the proof is finished.
\end{proof}

The average packet in (\ref{t1}) decreases with the increase of caching storage $M$, the wireless bandwidth $B$ and the slot duration $\tau$. Increasing caching storage can reduce the burden of the wireless spectrum to meet the demand of a given PLR target. Caching storage can be exchanged for scarce spectrum and time resources. Furthermore, when the packets are with different access popularity, we have the following proposition.
\begin{proposition}
For $R_b\rightarrow\infty$, $\sigma^2\rightarrow0$ and $q_j=\frac{M}{N}$, when the packets have different access popularity, the minimum achievable PLR can be obtained with the caching scheme  developed via the following optimization problem.
\begin{align}
&\min_{\mathbf{P}, \mathbf{Q}} ~~\mathcal{P}=\sum\limits_{i=1}^{ \binom{N}{M}}p_i\frac{1-\sum_{j=1}^Nq_{i,j}f_j}{1+(1-\sum_{j=1}^Nf_jq_{i,j})^{-1}\mathcal{Z}^{-1}_1(\bar{T})}\label{aa}\\
&s.t.~~\left\{
\begin{array}{l}
q_{i,j}\in\{0,1\};\\
0\leq p_i\leq1; \\
\sum_{j=1}^N q_{i,j}= M;\\
\sum_{i=1}^{\binom{N}{M}} p_i=1;\\
\sum_{i=1}^{\binom{N}{M}} q_{i,j}p_i=\frac{M}{N};\\
{{p_i\leq p_j,\forall j<i \in \{1,2,...,\binom{N}{M}\}}}.
\end{array}\right.
\end{align}
\end{proposition}

The objective function in (\ref{aa}) is derived from (\ref{with}) and (\ref{average}) by noting $\alpha_j=f_j$ when $q_j=\frac{M}{N}$. The programming to get the optimal caching scheme can be simply described in Algorithm 1. We then have the minimum achievable PLR via substituting the optimal $\mathbf{P}$ and $\mathbf{Q}$ into the objective function in (\ref{aa}).
\begin{table}
\label{table:pushing}
\begin{center}
\begin{tabular}{ll}
\hline
\multicolumn{2}{l}{\textbf{Algorithm 1:}  Programming for caching scheme}\\
\hline
1:& \textbf{Input}. $N, M,B,\tau,T, f_j, \forall j \in\{1,2,\cdots,N\}$.\\
2:& \textbf{Initialize}. $\mathbf{P}\leftarrow(0)_{N\times 1}$, $\mathbf{Q}\leftarrow(0)_{\binom{N}{M}\times N}$.\\
3:& $\binom{N}{M}$ schemes to select $M$ elements out of $\{l_1,l_2,\cdots,l_N\}$. \\
4:& \textbf{if} $l_j$ has been selected in the $i$-th scheme, $\forall i \in\{1,2,\cdots,\binom{N}{M}\}$\\
5:&  ~~~$q_{i,j}\leftarrow 1$\\
6:& \textbf{end if}\\
7:& Substitute $\mathbf{Q}$ into the problem and solve the linear optimization\\
&problem regarding $\mathbf{P}$ without the constrain $p_i\leq p_j,\forall j<i$.\\
8:& Rearrange the order of rows for $\mathbf{P}, \mathbf{Q}$ to satisfy $p_i\leq p_j,\forall j<i$.\\
9:& \textbf{Output}. Caching scheme $\mathbf{P}, \mathbf{Q}$.\\
\hline
\end{tabular}
\end{center}
\end{table}
\section{Numerical results and discussions}
We evaluate the performance of  the proposed system in this section. The BSs and users are located based on PPPs with intensities of  $\{\lambda_b,\lambda_u\}=\{\frac{100}{\pi500^2},\frac{2000}{\pi500^2}\}~\text{nodes/m}^2$. The wireless bandwidth $B=20$ MHz, the slot duration $\tau=0.5$ second and the transmit power $P=33$ dBm. We consider the interference-limited network ($\sigma^2=0$) and the path-loss $\beta=4$. Set the packet access popularity $f_{i}={\frac{1/i^{\gamma}}{\sum_{j=1}^{C}1/j^{\gamma}}}$ according to the Zipf distribution and $\gamma=0.8$. The packet library size $N=100$, the packet size $T=10$ Mbits, and the caching storage of each user $M=3$.
\begin{figure}[t]
\centering
\vspace{-0.15in}
\includegraphics[width=2.8in]{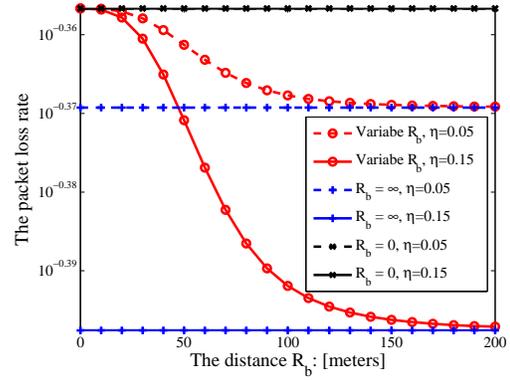}
\caption{The impact of $R_b$ on the PLR for the un-cached packets in the cache-enabled network.}
\label{R}
\end{figure}

Fig. \ref{R} demonstrates the impact of $R_b$ on the loss rate of the un-cached packets, larger $R_b$ implies more CSI are available for the cache-enabled users. Denote $\eta=\sum_{j=1}^N\alpha_jq_{i,j}$ for users in the $i$-th subset. Larger $\eta$ implies more interference can be cancelled with the cached packets for the users. It can be observed that the packet loss rate with complete interference cancellation ($R_b=\infty$) is much smaller that that without interference cancellation ($R_b=0$). It verifies the necessary and the efficiency of the interference cancellation proposed in this paper for the cache-enabled network. The packet loss rate decrease with the increase of $R_b$. The performance is approximately equal to that of complete interference cancellation ($R_b=\infty$) when $R_b=120$ ($180$) m for $\eta=0.05$ ($\eta=0.15$). Higher $\eta$ needs larger area of CSI (larger $R_b$) to achieve the approximate performance of the complete interference cancellation. On the other hand, lower packet loss rate can be achieved with larger $\eta$ because more interference is cancelled. So smaller area of CSI (smaller $R_b$) is needed to meet the given packet loss rate target when $\eta$ is larger.
\begin{figure}[t]
\centering
\includegraphics[width=2.8in]{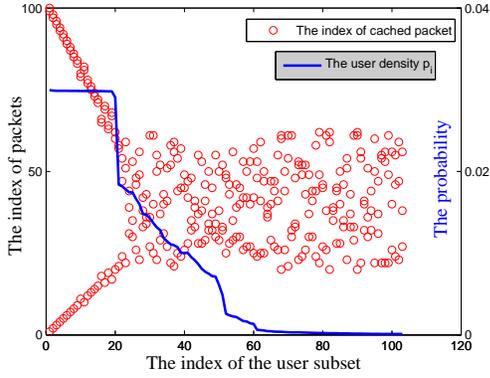}
\caption{The caching scheme based on Algorithm 1.}
\label{C}
\end{figure}

Fig. \ref{C} illustrates the optimal caching scheme which is obtained based on Algorithm 1. The probability ($p_i$) of users in the $i$-th subset and the packets cached in the corresponding users are illustrated in the figure.  The left and right ordinate are respectively  the index of the packets and the probability of users in each subsect. The red circle located at $(i,j)$ means the $j$-th packet is cached by the users in the $i$-th subset.  It can be observed that users in different subsets are with different densities. The users caching both the popular and unpopular packets are with higher density, yielding both the local caching gain and the interference cancellation gain. 
\begin{appendix}
Based on (\ref{sinr}) and (\ref{lossdefine}), we have
\begin{align}\label{esinr}
\mathcal{P}_l&=1\!-\!\mathbb{E}\Big[\mathbb{P}\Big(\text{SINR}\!\geq\!\bar{T}\Big)\Big]=1\!-\!\mathbb{E}_r\left[\mathbb{P}\left(h_{0,0}\!>\!r^\beta P^{-1}\bar{I}_c\bar{T}
                      \right)\right]\nonumber\\
                      &=1-\int_0^{{\infty}} \mathbb{P}\left[h_{0,0}>r^\beta P^{-1}\bar{I}_c\bar{T}
                      \right]f_{R}(r)\mathrm{d}{r}\nonumber\\
                      &=1-\Big(\int_0^{{R_b}}+\int_{R_b}^{{\infty}}\Big) \mathbb{P}\left[h_{0,0}>r^\beta P^{-1}\bar{I}_c\bar{T}
                      \right]f_{R}(r)\mathrm{d}{r}\nonumber\\
                      &\triangleq1-\mathcal{P}_{s,i}-\mathcal{P}_{b,i}.
\end{align}
Next, we will get $\mathcal{P}_{s,i}$ and $\mathcal{P}_{b,i}$ separately. Firstly, for $\mathcal{P}_{s,i}$, when $R\leq R_b$, 
the SINR in (\ref{sinr}) can be written as follows,
\begin{align}\label{sinrs}
    &\frac{P|h_{0,0}|^2d_{0,0}^{-\beta}}{\sum_{j\in\Phi_{b1}\odot R_b }{P|h_{j,0}|^2d_{j,0}^{-\beta}}+\sum_{k\in\Phi_{b2}\odot R }{P|h_{k,0}|^2d_{k,0}^{-\beta}}+\sigma^2}\nonumber\\
    &\triangleq\frac{P|h_{0,0}|^2d_{0,0}^{-\beta}}{I_{c,s}+I_{u,s}+\sigma^2}\triangleq\frac{P|h_{0,0}|^2d_{0,0}^{-\beta}}{\bar{I}_{c,s}}.
\end{align}
  Therefore, for $\mathcal{P}_{s,i}$, we have the follows,
\begin{align}\label{esinr}
\mathcal{P}_{s,i}&=\int_0^{R_b} \mathbb{P}\left[h_{0,0}>r^\beta P^{-1}\bar{I}_{c,s}\bar{T}
                      \right]f_{R}(r)\mathrm{d}{r},
\end{align}
where
\begin{align}\label{pgi0}
&\mathbb{P}\left[h_{0,0}>r^\beta P^{-1}\bar{I}_{c,s}\bar{T}\right]=\mathbb{E}_{\bar{I}_{c,s}}[e^{-r^\beta P^{-1}\bar{I}_{c,s}\bar{T}}]\nonumber\\
                         &=e^{-r^\beta P^{-1}\bar{T}\sigma^2}\mathbb{E}_{{I}_{c,s}}[e^{-r^\beta P^{-1}{I}_{c,s}\bar{T}}]\times\mathbb{E}_{{I}_{u,s}}[
                           e^{-r^\beta P^{-1}{I}_{u,s}\bar{T}}]\nonumber\\
                        &=e^{-r^\beta P^{-1}\bar{T}\sigma^2}\mathcal{L}_{I_{c,s}}\left[r^\beta P^{-1}\bar{T}\right]\times\mathcal{L}_{I_{u,s}}\left[r^\beta P^{-1}\bar{T}\right].
\end{align}
Interference $I_{c,s}$ is from BSs which are spatially distributed as PPP $\Phi_{b1}$ and outside the circle centered at the origin with radius $R_b$. So the Laplace transform $\mathcal{L}_{I_{c,s}}[r^\beta P^{-1}\bar{T}]$ is
\begin{align}\label{laplace1}
&\mathcal{L}_{I_{c,s}}\left[r^\beta P^{-1}\bar{T}\right]=\mathbb{E}_{I_{c,s}}\left[\text{exp}\left({-r^\beta P^{-1}\bar{T}I_c}\right)\right]\nonumber\\
&=\mathbb{E}_{\Phi_{b1},\{|h_{j,0}|^2\}}\Big[\text{exp}\Big({-r^\beta P^{-1}\bar{T}{\sum_{j\in\Phi_{b1} \odot R_b }{P|h_{j,0}|^2d_{j,0}^{-\beta}}}}\Big)\Big]\nonumber\\
&=\mathbb{E}_{\Phi_{b1},\{|h_{j,0}|^2\}}\Big[\prod_{j\in\Phi_{b1}\odot R_b}\text{exp}\left({-r^\beta\bar{T}|h_{j,0}|^2d_{j,0}^{-\beta}}\right)\Big]\nonumber\\
&=\mathbb{E}_{\Phi_{b1}}\Big[\prod_{j\in\Phi_{b1}\odot R_b}\mathbb{E}_{\{|h_{j,0}|^2\}}\Big[\text{exp}\left({-r^\beta\bar{T}|h_{j,0}|^2d_{j,0}^{-\beta}}\right)\Big]\Big]\nonumber\\
&\stackrel{(a)}{=}\mathbb{E}_{\Phi_{b1}}\Big[\prod_{j\in\Phi_{b1}\odot R_b}  \frac{1}{1+r^\beta\bar{T}d_{j,0}^{-\beta}}\Big]\nonumber\\
&=\text{exp}\Big[{-2\pi\lambda_{b1}\int_{R_b}^\infty\Big(1-\frac{1}{1+r^\beta\bar{T}v^{-\beta}}\Big)v\mathrm{d}{v}}\Big]\nonumber
\end{align}
\begin{align}
&=\text{exp}\Big[{-2\pi\lambda_{b1}\mathlarger{\int}_{R_b}^\infty\frac{v}{1+(r^\beta\bar{T})^{-1}v^{\beta}}\mathrm{d}{v}}\Big]\nonumber\\
&\stackrel{(b)}{=}\text{exp}\Big[{-\pi\lambda_{b1}r^2{\bar{T}}^\frac{2}{\beta}\mathlarger{\int}_{[(\frac{r}{R_b})^\beta{\bar{T}}]^{-\frac{2}{\beta}}}^\infty({1+u^{\frac{\beta}{2}}})^{-1}\mathrm{d}{u}}\Big]\nonumber\\
&=\text{exp}\Big\{{-\pi\lambda_{b1}r^2\mathcal{Z}_1\Big[\Big(\frac{r}{R_b}\Big)^\beta\bar{T}\Big]}\Big\}.
\end{align}
By noting the Rayleigh channel fading we have Step (a). Using change of variables with $u=[r^\beta\bar{T}]^{-\frac{2}{\beta}}v^2$, we have Step (b). Denote $\mathcal{Z}_1(x)\triangleq\frac{2\bar{T}}{\beta-2}{}_2F_1[1,1-\frac{2}{\beta};2-\frac{2}{\beta};-x]$,
where ${}_2F_1[\cdot]$ is the Gauss hypergeometric function. Furthermore, interference $I_{u,s}$ comes from the BSs which are spatially distributed as PPP $\Phi_{b2}$ and outside the circle centered at the origin with radius $R$. With the same approach of (\ref{laplace1}), the simplified derivation for the Laplace transform $\mathcal{L}_{I_{u,s}}[r^\beta P^{-1}\bar{T}]$ is
\begin{align}\label{laplace2}
&\mathcal{L}_{I_{u,s}}\left[r^\beta P^{-1}\bar{T}\right]=\mathbb{E}_{I_{u,s}}\left[\text{exp}({-r^\beta P^{-1}\bar{T}I_u})\right]\nonumber\\
&=\mathbb{E}_{\Phi_{b2},\{|h_{k,0}|^2\}}\Big[\text{exp}\Big({-r^\beta P^{-1}\bar{T}{\sum_{k\in\Phi_{b2} \odot R }{P|h_{k,0}|^2d_{k,0}^{-\beta}}}}\Big)\Big]\nonumber\\
&=\text{exp}[{-\pi\lambda_{b2}r^2\mathcal{Z}_1(\bar{T})}].
\end{align}
Substituting (\ref{pgi0}), (\ref{laplace1}) and (\ref{laplace2}) into (\ref{esinr}), we have $\mathcal{P}_{s,i}$ in (\ref{long}).

Similarly, $\mathcal{P}_{b,i}$ can be obtained with the same approach of deriving $\mathcal{P}_{s,i}$, and the proof is omitted here for the limited spaces. We then get Theorem \ref{theorem1} and the proof is finished.  \hfill\ensuremath{\blacksquare}
\end{appendix}
\bibliographystyle{IEEEtran}
\bibliography{paperinter}

\end{document}